\documentclass[12pt]{article}

\usepackage[margin=2.5cm]{geometry}
\usepackage{amsthm,amsmath,amsfonts,amssymb}
\usepackage{graphicx,float}
\usepackage{multirow,setspace}
\usepackage[authoryear]{natbib}
\usepackage{enumerate}
\usepackage{caption}
\usepackage{subcaption}
\usepackage{bm}
\usepackage{color}
\usepackage{url}
\usepackage[colorlinks,citecolor=blue,urlcolor=blue]{hyperref}
\usepackage[format=plain, labelfont={bf,it}, textfont=it]{caption}
\theoremstyle{plain}
\newtheorem{axiom}{Axiom}
\newtheorem{claim}[axiom]{Claim}
\newtheorem{theorem}{Theorem}[section]

\theoremstyle{remark}

\newtheorem{remark}[theorem]{Remark}
\title{\vspace{-40pt}Doubly ranked tests of location for grouped functional data\vspace{-10pt}}
\author{Mark J. Meyer \\ Department of Mathematics and Statistics, Georgetown University \\ Washington, DC USA}
\date{}

\begin{document}

	\maketitle
	
\begin{abstract}
Nonparametric tests for functional data are a challenging class of tests to work with because of the potentially high dimensional nature of the data. One of the main challenges for considering rank-based tests, like the Mann-Whitney or Wilcoxon Rank Sum tests (MWW), is that the unit of observation is typically a curve. Thus any rank-based test must consider ways of ranking curves. While several procedures, including depth-based methods, have recently been used to create scores for rank-based tests, these scores are not constructed under the null and often introduce additional, uncontrolled for variability. We therefore reconsider the problem of rank-based tests for functional data and develop an alternative approach that incorporates the null hypothesis throughout. Our approach first ranks realizations from the curves at each measurement occurrence, then calculates a summary statistic for the ranks of each subject, and finally re-ranks the summary statistic in a procedure we refer to as a doubly ranked test. We propose two summaries for the middle step: a sufficient statistic and the average rank. As we demonstrate, doubly rank tests are more powerful while maintaining ideal type I error in the two sample, MWW setting. We also extend our framework to more than two samples, developing a Kruskal-Wallis test for functional data which exhibits good test characteristics as well. Finally, we illustrate the use of doubly ranked tests in functional data contexts from material science, climatology, and public health policy.
\end{abstract}

	\section{Introduction}
	\label{s:intro}
	
	Expanded research in functional data analysis has produced a variety of tests couched in the framework of classic nonparametric tests including work by \cite{Hall2007, Lopez2009, Lopez2010, Chak2015, PomannStaicu2016, Lopez2017, Abramowicz2018, Berrett2021} and \cite{Melendez2021}. These various authors have built out hypothesis testing procedures for functional versions of Mann-Whitney or Wilcoxon Rank Sum tests, Anderson-Darling tests, and permutation tests, among others. The primary feature of functional data that makes the development of such nonparametric tests challenging is that the unit of observation is typically a curve. For example, $X_{g,1}(s), \ldots, X_{g,n_g}(s)$ are a sample of $n_g$ curves from groups $g = 1, \ldots, G$. Each curve is typically measured on a fine grid, $\mathcal{S} = \{s : s = s_1, \ldots, s_S\}$ for a total of $S$ measurements. Any rank-based procedure must employ a way to rank curves directly or at least rank summary metrics of each subject's curve. While there are higher dimensional forms of functional data, we limit our current examination to curves. 
	
	Depth is a popular and useful concept for ranking functions and many authors have utilized depth for Mann-Whitney-Wilcoxon (MWW) type tests \citep{Lopez2009,Lopez2010,Lopez2017}, functional boxplots \citep{Sun2011, Genton2014, Wrobel2016}, outlier detection \citep{SunGenton2012, Dai2018, Huang2019, Dai2020, AlemanGomez2022}, and ranking curves \citep{Lopez2009, Lopez2011, Sun2012, Sguera2021}. The basic idea behind depth is to characterize the proportion of time a curve is bounded above and below by another curve in the sample. The resulting scores, when ordered, rank the curves from the inside out with the smallest values representing the extremes, the maximum or minimum, and the largest value representing the median \citep{Lopez2009, Sun2012}. These scores are effectively a measure of the distance a curve is from the middle. MWW-type tests using depth to test for a difference between two groups must construct a third, artificial reference group  \citep{Lopez2009,Lopez2010,Lopez2017}. Depth is first determined by group, then the test procedure counts the number of depth values in each group that are larger than the depth in the reference group. Code is publicly available for the depth-based MWW test, see \cite{Lopez2017}.
	
	Alternative approaches include a spatial sign-based statistic by \cite{Chak2015} and a random projections based approach by \cite{Melendez2021}. In the former,  \cite{Chak2015} develop an MWW-type statistic using the spatial sign for infinite-dimensional data that weakly convergences to a Gaussian distribution under limited assumptions.  The spatial sign is used as a distance metric to construct the test, taking pairwise distances between groups with data assumed to arise from a random variable defined on a separable Hilbert space. The latter alternative approach, a random projection-based test by \cite{Melendez2021}, maps the points from the high dimensional functional space to a randomly chosen low-dimensional space defined using Brownian motion. This approach has application in parametric tests for grouped functional as well, see for example \cite{Cuesta2010}. The random projections-based approach effectively generates scores for each subject via a random basis function and an integral approximation of the subject-specific curve. It then treats these scores as the data in a traditional MWW test, thus the scores are ranked and then summed by group. Neither of these alternative methods have publicly available code.
	
	Both depth- and random projection-based approaches perform well in larger samples \citep{Lopez2009,Lopez2010,Melendez2021}. However, each method also introduces an additional and uncontrolled for source of variation: the random artificial reference group for the depth-based approach and the brownian motion-based basis functions for the random projections. While depth is useful for ranking curves particularly when identifying medians and outliers, it does not produce a score that ranks in a traditional fashion, from smallest to largest. When using random projections, the integration is performed before constructing ranks which ignores the fact that rank may be dynamic over time. To handle the curve ranking problem, both procedures rely on subject-specific summary scores, although neither constructs their scores under the null. Thus, the tests are only conducted under the null when comparing the scores between groups. 
	
	The spatial sign-based approach differs from the depth- and random projection-based methods in that, for finite-dimensional approximations of functional data, it effectively creates a time-point specific score instead of subject-specific scores over time.
	\cite{Chak2015} show that the spatial sign-based test performs well empirically when the underlying data is Gaussian with a relatively small sample. The test does consider the MWW null throughout its construction, but appears to be underpowered when the data is non-Gaussian. While all of these methods are formulated for the MWW or two group case, to our knowledge no prior work considers the Kruskal-Wallis (KW) or three or more group case. Thus, with respect to these authors, we believe that a further examination of this problem is warranted. 
	
	In this manuscript, we re-examine rank-based test settings for comparing groups of functional data, where the data are curves. We propose a general testing procedure working within the MWW and KW assumptions, developing our tests under the relevant null hypotheses at all stages---a general overview of MWW and KW tests as well as functional data is given in Section~\ref{s:statframe}. Our approach first constructs ranks of either the raw data or preprocessed functional data at each discretely measured time point giving every subject a curve of ranks. We then construct a summary statistic for each subject using one of two choices: a sufficient statistic or the average rank (Section~\ref{s:drts}). The former is derived from the exact and approximate distribution of each subject's $r$th order statistic at a given time point under $H_0$ . The resulting approximate distribution can be shown to be an exponential family member from which we obtain sufficient statistics for each subject's value of $r$. We also examine the properties of the average rank under $H_0$. The summaries resulting from either the sufficient statistic or average rank become our new data, or rank-based scores, which are then ranked and analyzed in the usual way by the relevant tests. Because the MWW and KW tests ultimately re-rank our rank-based scores, we refer to our testing procedures as doubly ranked tests.
	
	We empirically demonstrate the power of doubly ranked tests in Section~\ref{s:sim}. Our approach can accommodate curves with a potentially large number of measurements and groups of size $G \geq 2$. Further, we show that doubly ranked tests perform well in both small and larger sample sizes. We illustrate their use with an analysis of data from several different studies across different substantive fields (Section~\ref{s:app}). The illustrations include data on resin viscosity, Canadian weather patterns, and changes in driving requests in the United States prior to and just after the initiation of COVID-19 policy measures in various states. While we formulate doubly ranked tests in the functional data context with data that is preprocessed using functional principal components, the tests are applicable to other high dimensional data testing settings. We discuss this, and our work in general, in Section~\ref{s:disc}.
		
	\section{Statistical Framework}
	\label{s:statframe}

	Because our method blends several statistical frameworks, we now briefly describe each beginning with the MWW and KW tests in their scalar forms. We then discuss some general concepts for functional data. Functional data analysis is a large area of research. Review articles by \cite{Morris2015}, \cite{Wang2016}, and \cite{Greven2017}, for example, or the text by \cite{Ramsay2005} will provide more in-depth discussions of this topic.
	
	\subsection{Mann-Whitney-Wilcoxon}
	
	Let $X_1, \ldots, X_{n_1}$ be observed data from group 1, $Y_1, \ldots, Y_{n_2}$ be observed data from group 2, and $n = n_1 + n_2$. Assuming the groups are independent, the null and alternative hypotheses can have one of two forms. The first tests for stochastic ordering where the null is $H_0: F_X(c) = F_Y(c)$ and the alternative is $H_A: F_X(c) \leq F_Y(c)$ for the CDFs $F_X(c)$ and $F_Y(c)$ in each group evaluated at some real valued constant $c$. Assuming these distributions are the same up to a location shift $\Delta$, i.e. $F_X(c) = F_Y(c-\Delta)$, gives a second form of the hypotheses which is $H_0: \Delta = 0$ and $H_A: \Delta \neq 0$. 
	
%
	The test statistic to test either form of $H_0$ is $T^+ = T - \frac{n_2(n_2 + 1)}{2}$ where $T = \sum_{i=1}^{n_2} R[Y_i]$ for $R[Y_i]$ equal to the rank of $Y_i$ among the combined samples, i.e. among $X_1, \ldots, X_{n_1},$ $Y_1, \ldots, Y_{n_2}$. The sum of the ranks, $T$, is the test statistic for the Wilcoxon rank sum test \citep{Wilcoxon1945,Kloke2015}. $T^+$ is also equivalent to the test statistic described by \cite{MannWhitney1947}. Consequently, we refer to this test as the MWW test.
	
	The exact distribution of $T$ is non-standard but can be easily computed by standard statistical software for samples up to $n = 50$, provided there are no tied ranks. For larger samples and in the presence of ties, a normal approximation to the distribution of the ranks is used instead, see the \texttt{wilcox.test()} documentation in the \texttt{stats} package in \texttt{R} \citep{Stats2022}. The test statistic based off of the normal approximation is
	\begin{align*}
		T_z = \left[T - \frac{n_2(n+1)}{2}\right]\bigg/ \sqrt{\frac{n_1 n_2 (n+1)}{12}},
	\end{align*}
	which has a standard normal distribution.

	\subsection{Kruskal-Wallis}
	
	The KW test generalizes the MWW to three or more groups \citep{KruskalWallis1952,Kloke2015}. Let $G$ denote the total number of independent groups with data $X_{g,1}, \ldots, X_{g,n_g}$ for $g = 1, \ldots, G$. The KW test assumes the distributions of the data are the same except possibly for some location parameter or parameter(s). The null hypothesis for the KW test generalizes the second MWW null where groups differ only up to some set of location shifts, $\Delta_g$. The null then has the form $H_0 : \Delta_1 = \ldots = \Delta_G$ with the alternative hypothesis, $H_A : \Delta_g \neq \Delta_{g'}$ for some $g \neq g'$. 
	
	For $n = \sum_{g=1}^G n_g$, the test statistic to test $H_0$ is
	\begin{align*}
		H = \frac{12}{n(n+1)} \sum_{g = 1}^G n_{g} \left( \bar{R}_g - \frac{n+1}{2} \right)^2,
	\end{align*}
	where $\bar{R}_g$ is the average rank in group $g$, i.e. $\bar{R}_g = n_g^{-1} \sum_{i=1}^{n_g} R[X_{g,i}]$. Under $H_0$, $H$ is distribution-free. Tables for its exact distribution do exist, see for example \cite{Hollander1999}, Chapter 6. $H$ is asymptotically $\chi_{G-1}^2$ under $H_0$ as well, see \cite{KruskalWallis1952} or \cite{Kloke2015} which is what the the \texttt{stats} package in \texttt{R} uses when \texttt{kruskal.test()} is called \citep{Stats2022}.
	
	\subsection{Functional Data}

	\begin{figure}
		\centering
		\includegraphics[width = 2.1in, height = 2.1in]{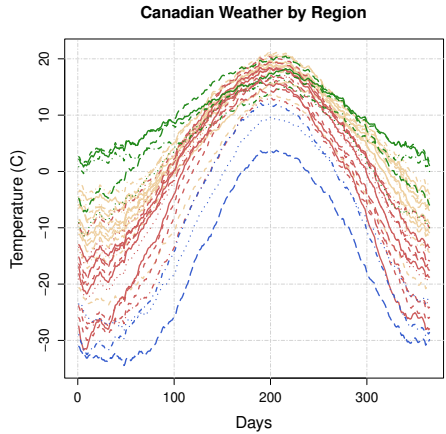}
		\caption{Blue curves denote the stations in the Arctic region, tan curves indicate stations the Atlantic region, red curves for stations in the Continental region, and green curves are the Pacific region.\label{f:cwtemp}}
	\end{figure}
	
	Functional data are data observations whose units are commonly curves, often measured over time but they can be measured over location as well. These curves are typically sampled at a high frequency over a pre-defined interval. The data comes sampled on a grid, $\{s : s = s_1, \ldots, s_S\}$ for example as we note in Section~\ref{s:intro}. The data used for analysis may be the raw values themselves, $X_1(s), \ldots, X_{n_1}(s)$, or they may be preprocessed and projected into a well-behaved, often lower dimensional space using a known basis function. The idea behind preprocessing is to model the underlying function from which the raw data values were observed and then analyze the functions themselves. We denote the preprocessed data with $\tilde{X}_1(s), \ldots, \tilde{X}_{n_1}(s)$. There are many ways to preprocess functional data including functional principal components, wavelet-based techniques, and Fourier transforms \citep{Wang2016}. Our method is compatible with either raw or pre-processed data, but for this work we will preprocess functional data using a sandwich smoother-based approach referred to as Fast Covariance Estimation or FACE \citep{Xiao2013, Xiao2016}. FACE is a nonparametric approach for constructing functional principal components based on a set of curves with a potentially large sampling density. It can be implemented using the \texttt{refund} package in \texttt{R} \citep{Refund2022}. As an example, Figure~\ref{f:cwtemp} contains the preprocessed (using FACE and retaining 99\% of the variability) daily temperature curves from the Canadian Weather dataset \citep{Ramsay2005,Ramsay2009}. The dataset contains daily temperature and precipitation measurements, averaged over the years 1960 to 1994 from 35 different sites across Canada. We analyze this data in Section~\ref{s:app}.
	
	\section{Doubly Ranked Tests}
	\label{s:drts}
	
	
	Suppose we obtain a sample of curves from two or more groups: $X_{g,1}(s), \ldots, X_{g, n_g}(s)$ for $g = 1, \ldots, G$, $s = s_1, \ldots, s_S$, and $n_1 + \ldots + n_g + \ldots + n_G = n$. Measurements are discretely sampled on the same grid, $\mathcal{S} = \{s : s = s_1, \ldots, s_S\}$, and come from a continuous process. Under the null, we assume no difference in the curves between the $G$ groups across all $s$. To avoid having to make a distributional assumption on the processes that generate $X_{g,i}(s)$, we evaluate this null by analyzing the ranks of $X_{g,i}(s)$ or $\tilde{X}_{g,i}(s)$ at each measurement occurrence $s$, i.e. time point or location depending on what the function is measured over. Thus, we first rank either the raw or preprocessed data at each measurement occurrence, ignoring group assignment. This results in a curve of ranks for each subject which we denote as $z_{g,i}(s)$. If there is no difference in the curves across $s$ under the null, then we should also observe no difference in the ranks across $s$. Assuming each subject's curve has a true order in the sample, we can consider $\{z_{g,i}(s_1), \ldots, z_{g,i}(s_S)\}$ as a set of draws of size $S$ from the distribution describing that true order.
	
	As in the univariate tests, the ranks are randomly assigned under $H_0$. Thus, for each subject and at each measurement occurrence, $z_{g,i}(s) \stackrel{H_0}{\sim} \mathcal{U}\{1,n\}$ where $\mathcal{U}$ denotes the discrete uniform distribution. Under $H_0$, we make several observations:
	
	\begin{remark}\label{rm:id}
		The sample  $\{z_{g,i}(s_1), \ldots, z_{g,i}(s_S)\}$ is identically distributed.
	\end{remark}

	\begin{remark}\label{rm:ev}
		For a single value of $s$, $E[z_{g,i}(s)] = \frac{n+1}{2} < \infty$ for $n < \infty$.
	\end{remark}

	\begin{remark}\label{rm:var}
		For a single value of $s$, $Var[z_{g,i}(s)] = \frac{n^2 - 1}{12} < \infty$ for $n < \infty$.
	\end{remark}		
	 After the first stage of ranking, we construct a summary statistic for each curve of ranks and re-rank the summaries for analysis. Assuming the summaries are reasonable representations of the subject's order within the sample, no difference in the curves and subsequent ranks over $s$ between the groups implies no difference in the summarized ranks between the groups. The doubly ranked test inherits the null and alternatives from the MWW and KW tests since the last step analyzes the summaries using either the MWW or KW. From this, we can construct a more formal set of hypotheses.
	 
	 Via Remark~\ref{rm:id}, it is reasonable to assume the distributions of the summary statistics are the same up to some potential location shift. Let $t[z_{g,i}(\mathcal{S})]$ denote the summary statistic for a subject in group $g$ and suppose it is a realization of the random variable $T_g$. The doubly ranked MWW test then has the null and alternative of $H_0: \theta = 0$ and $H_0: \theta \neq 0$ where $\theta$ is a parameter such that $F_{T_1}(c) = F_{T_2}(c-\theta)$ for some real valued constant $c$. When $G \geq 3$, the null and alternative hypotheses for the doubly ranked KW test are $H_0 : \theta_1 = \ldots = \theta_G$ with the alternative hypothesis, $H_A : \theta_g \neq \theta_{g'}$ for some $g \neq g'$ where $\theta_g$ is the location parameter for the distribution $F_{T_g}(c)$. In other words, under $H_0$ we assume there is no difference in the location parameters of the distributions of $t[z_{g,i}(\mathcal{S})]$ between the groups. The alternative is that the there is a difference, when $G = 2$, or there is at least one difference, when $G \geq 3$. 
	
	Our testing procedure for grouped functional data proceeds as follows: 1) preprocess the data using FACE, ignoring group assignment, to obtain $\tilde{X}_{g,1}(s), \ldots, \tilde{X}_{g, n_g}(s)$; 2) generate ranks at each time point $s$, again ignoring group assignment; 3) calculate a $t[z_{g,i}(\mathcal{S})]$; 4) perform the relevant test (MWW or KW), depending on $G$. The doubly ranked MWW test statistic is
	\begin{align*}
		T_{DR}^+ = \sum_{j=1}^{n_2} R\left( t\left[z_{2,j}(\mathcal{S})\right] \right) - \frac{n_2(n_2 + 1)}{2}.
	\end{align*}
	For $G \geq 3$, the doubly ranked KW test statistic is
	\begin{align*}
		&H_{DR} = \frac{12}{n(n+1)} \sum_{g = 1}^G n_{g} \left[ \overline{R\left( t\left[z_{g,i}(\mathcal{S})\right] \right)} - \frac{n+1}{2}\right]^2,
	\end{align*}
	where $\overline{R\left( t\left[z_{g,i}(\mathcal{S})\right] \right)} = \frac{1}{n_g} \sum_{i=1}^{n_g} R\left( t\left[z_{g,i}(\mathcal{S})\right] \right)$.The statistic $t\left[z_{g,i}(\mathcal{S})\right]$ is a univariate score for each subject which we use as the ``new'' data in each backend test. Thus, the distributional results from Section~\ref{s:statframe} hold for both the doubly ranked MWW and KW tests. That is, $T_{DR}^+$ is approximately normal while $H_{DR}$ is asymptotically $\chi_{G-1}^2$. 
	
	 To obtain the summary statistic, $t[z_{g,i}(\mathcal{S})]$, we consider two approaches: a sufficient statistic derived from the distribution of the $r$th order statistic and the sample average of the ranks.

	\subsection{Sufficient Statistic}
	
	The summarizing step is a dimensionality reduction step which can be accomplished using a statistic that is sufficient for a curve's order within the sample. Let $Z_{(r)}$ denote the random variable describing the sampling distribution of the $r$th order statistic at time $s$. We let $z$ generically denote an observed rank for ease of derivation. The exact PMF for $Z_{(r)}$ is given by
\begin{align}
	P\bigg[ Z_{(r)} = z &\bigg] =  \frac{\Gamma(n+1)}{\Gamma(r)\Gamma(n - r + 1)} \label{eq:pzr}\\
	&\times \left[ \int_{\frac{z}{n}-\frac{1}{n}}^{\frac{z}{n}} t^{r-1}\left(1-t\right)^{n-r } dt \right].\nonumber
\end{align}
A full derivation of $P\left[ Z_{(r)} = z \right]$ is in Section 1 of Supplement A \citep{DRTSup2024A}.

The exact distribution is non-standard and difficult to work with. However, we can approximate the integral using the midpoint rule and selecting a single interval, $[\frac{z}{n} - \frac{1}{n}, \frac{z}{n}]$, with a midpoint at $\frac{z}{n} - \frac{1}{2n}$. 
The mass function in~\eqref{eq:pzr} can then be approximated by
\begin{align}
	P\bigg[ Z_{(r)} = z &\bigg] \approx \frac{\Gamma(n+1)}{\Gamma(r)\Gamma(n - r + 1)} \frac{1}{n}  \label{eq:pmf} \\
	& \times  \left(\frac{z}{n} - \frac{1}{2n}\right)^{r-1} \left(1-\frac{z}{n} + \frac{1}{2n}\right)^{n-r}, \nonumber
\end{align}
with the approximation improving as $n$ increases. We observe this graphically in Section 1 of Supplement A \citep{DRTSup2024A}. When $r = \frac{n+1}{2}$, the expected rank under $H_0$, the approximation is quite good.


\begin{claim}\label{cl:ef}
The distribution described by the mass function in Equation~\ref{eq:pmf} is an exponential family member.
\end{claim}

\begin{proof}
The proof of Claim~\ref{cl:ef} is in Appendix \ref{appx:ef}.
\end{proof}

Since Equation~(\ref{eq:pmf}) is an exponential family, $t(z) = \log\left[ \left(\frac{z}{n}- \frac{1}{2n}\right)\bigg/\left(1-\frac{z}{n} + \frac{1}{2n}\right) \right]$ is a sufficient statistic for $r$. Noting that, under $H_0$, $z \sim \mathcal{U}\{1,n\}$, this statistic has a useful property when applied to the sample of ranks. 

\begin{claim}\label{cl:ssm}
Under $H_0$, $E[(t(z)] = 0$.
\end{claim}

\begin{proof}
The proof of Claim~\ref{cl:ssm} is in Appendix \ref{appx:ev}.
\end{proof}

Given a sample of ranks over $s$ for subject $i$ in group $g$, $r$ could be considered the true order of subject $i$'s curve within the sample, under $H_0$. Thus we obtain our first approach for summarizing $z_{g,i}(s)$ over $s$:
\begin{align}
	&t_u\left[z_{g,i}(\mathcal{S})\right] = \frac{1}{S} \sum_{s=1}^S\log\left[ \frac{\frac{z_{g,i}(s)}{n}- \frac{1}{2n}}{1-\frac{z_{g,i}(s)}{n} + \frac{1}{2n}} \right]. \label{eq:suff}
\end{align}

	\begin{claim}\label{cl:suff} 
		Under $H_0$, $E\left[t_u\left\{z_{g,i}(\mathcal{S})\right\}\right] = 0$
	\end{claim}

	\begin{proof}
		 Via Claim~\ref{cl:ssm} and Remark~\ref{rm:id}, Claim~\ref{cl:suff} is straightforward to show after applying $E[\cdot]$ to the summation in Equation~\eqref{eq:suff}.
	\end{proof}

	The sufficient statistic, $t[z_{g,i}(s)]$, equals zero only when $z_{g,i}(s) = \frac{n+1}{2}$, which is the expected value of $z_{g,i}(s)$ via Remark~\ref{rm:ev}. When $S = 1$, the doubly ranked tests will be equivalent to their univariate analogue, i.e. the standard MWW or KW test, when using this summary since $t(z)$ is a one-to-one transformation of $z$.

	\subsection{Average Rank}

	Define the statistic $t_a\left[z_{g,i}(\mathcal{S})\right] = \frac{1}{S}\sum_{k=1}^S z_{g,i}(s_k)$ which takes the sample average of the observed ranks over the grid $\mathcal{S}$. Because the average is taken over time (or location) for a given subject and not over $n$, we now establish several simple results for $t_a\left[z_{g,i}(\mathcal{S})\right]$ under $H_0$.
		
	
	\begin{claim}\label{cl:arm} 
		$E\left[t_a\left\{z_{g,i}(\mathcal{S})\right\}\right] = \frac{n + 1}{2}$.
	\end{claim}

	\begin{proof}
		 Given Remarks~\ref{rm:ev} and~\ref{rm:id}, Claim~\ref{cl:arm} is straight-forward to establish by applying $E[\cdot]$ to the sum.
	\end{proof}
	
	An immediate result of Claim~\ref{cl:arm} is that $t_a\left[z_{g,i}(\mathcal{S})\right]$ is unbiased for the mean of the distribution of ranks under $H_0$. If we additionally assume that the  $\{z_{g,i}(s_1), \ldots, z_{g,i}(s_S)\}$ are independent as well as identically distributed, we have the following:
	
	\begin{claim}\label{cl:arc} 
		As $S \rightarrow \infty$, $t_a\left\{z_{g,i}(\mathcal{S})\right\} \stackrel{p}{\longrightarrow} \frac{n + 1}{2}$.
	\end{claim}
	
	\begin{proof}
		Claim~\ref{cl:arc} results from the law of large numbers under Remarks~\ref{rm:id} to~\ref{rm:var} and assuming the ranks are independently assigned over $\mathcal{S}$ under $H_0$.
	\end{proof}

	The assumption of independence implies that the ranks at time $s_k$ are assigned independently from the ranks at time $s_{k'}$ for $k \neq k'$. While this assumption is strong for functional data, Claim~\ref{cl:arc} only assumes it on the ranks, not the data, and only under the null. In the absence of the independence assumption, Claim~\ref{cl:arm} will still hold. When $S = 1$, the doubly ranked tests will also be equivalent to their univariate analogue when using the average rank.
	
	\subsection{Implementation}

Implementation of our method is done through the \texttt{R} package \texttt{runDRT} which is publicly available on \texttt{CRAN}. The code can implement both summary statistics though it defaults to the sufficient statistic, see the package documentation for details \citep{runDRT2024}.

	\section{Simulation Study}
	\label{s:sim}
	
	In our empirical study, we assume a balanced design with two groups of size $n_1 = n_2 = 10, 25$ and 50 in the MWW case and a third group in the KW case with $n_3$ defined similarly. The total sample sizes are $n = 20, 50,$ and 100 (MWW) and $n = 30, 75,$ and 150 (KW). To simulate the functional data, we first use the approach employed by \cite{Chak2015} who generate data in the MWW setting using a Karhunen-Lo\`eve expansion:
	\begin{align*}
		X_{1,i}(s) = \sum_{k=1}^K \sqrt{2} \left[(k - 0.5)\pi \right]^{-1} Z_{ik} \sin\left[(k - 0.5)\pi s\right],
	\end{align*}
	for those in group 1; \cite{Berrett2021} also take a similar approach. The value $Z_{ik}$ is a random variable which we generate from one of two distributions: Gaussian or $t_2$ (a $t$ distribution with 2 degrees of freedom). The former induces a Gaussian process while latter induces a $t$-process. Formally, the value of $K$ could be infinite. In simulation, we take it to be large, $K = 1000$ basis functions. To generate the data in the remaining group(s), we use
	\begin{align*}
		&X_{g,j}(s) = \\
		&\mu(s) + \sum_{k=1}^K \sqrt{2} \left[(k - 0.5)\pi \right]^{-1} Z_{jk} \sin\left[(k - 0.5)\pi s\right],
	\end{align*}
	where $\mu(s)$ is a time-dependent function for $s \in [0,1]$, $g = 2$ for the MWW case, and $g = 2, 3$ for the KW setting. This generative process does not include measurement error so we additionally consider two measurement error settings, each of the form $Y_{1,i}(s) = X_{1,i}(s) + \epsilon_i(s)$ and $Y_{g,j}(s) = X_{g,j}(s) + \epsilon_j(s)$. The error term can be white noise, $\epsilon_i(s) \stackrel{iid}{\sim} N(0,1)$, or auto-correlated noise, $\boldsymbol{\epsilon}_i = \{\epsilon_i(s_1), \ldots, \epsilon_i(s_S)\}' \sim N(0, \Sigma)$ where $\Sigma$ is an AR(1) correlation matrix with $\rho = 0.5$ and variance equal to 1.

	
	
	We consider three different functional forms for $\mu(s)$: 
	\begin{align*}
		\mu_1(s) = \xi s,\ &\mu_2(s) = \xi4s(1-s), \text{ and}\\
		\mu_3(s) = \xi m^{-1} &B(2, 6)^{-1} s^{2-1} (1-s)^{6-1},
	\end{align*}
	for $m$ equal to $\max_s B(2, 6)^{-1} s^{2-1}(1-s)^{6-1}$ and $B(2,6)$ equal to the Beta function evaluated at 2 and 6. The parameter $\xi$ is set to zero to evaluate Type I Error and incremented up to 3, by steps of size 0.12, to evaluate power. The length of the sampling grid, $s$, varies from sparser to denser: $S = 40, 120,$ and 360 and is taken to be equally spaced over the unit interval. We limit our direct comparison to the depth-based method, since it alone has publicly available code. However, we note that the forms $\mu_1(s)$ and $\mu_2(s)$ are used by \cite{Chak2015}. All curves are preprocessed using FACE, retaining 99\% of the variability. Code to generate the empirical studies is in Supplement B \citep{DRTSup2024B}.


\begin{table}
	\centering
	\caption{Type I Error for doubly ranked Mann-Whitney-Wilcoxon tests by $S$, distribution of $Z_k$, and $n$ under AR(1) noise. Each value is based on 10000 simulated datasets. Values closest to nominal are bolded. G stands for Gaussian, T for $t_2$, Avg. for the average rank, and Suff. for the sufficient statistic. \label{t:mwwe}}
	\begin{tabular}{lllccc}
		  \hline
		  \multirow{2}{*}{$S$} & \multirow{2}{*}{$Z_k$} & \multirow{2}{*}{$n$} &  \multicolumn{2}{c}{Doubly Ranked} & \multirow{2}{*}{Depth-based}  \\
		 \cline{4-5}
		  &  &  & Avg. & Suff. & \\ 
		  \hline
		 40 & G  & 20 & \bf 0.046 & 0.044 & 0.030 \\ 
		 & &  50 & 0.052 & \bf 0.051 & 0.044 \\ 
		 & & 100 & \bf 0.050 & \bf 0.050 & 0.048 \\ 
		\cline{3-6}
		 & T & 20 & 0.045 & \bf 0.046 & 0.033 \\ 
		 & &  50 & 0.049 & \bf 0.050 & 0.042 \\ 
		 & & 100 & 0.049 & \bf 0.050 & 0.045 \\ 
		 \cline{2-6}
		 120 & G  & 20 & \bf 0.044 & \bf 0.044 & 0.034 \\ 
		 & &  50 & \bf 0.052 & \bf 0.052 &  0.042 \\ 
		 & & 100 & \bf 0.050 & 0.049 & 0.048  \\ 
		\cline{3-6}
		 & T & 20 & 0.044 & \bf 0.045 & 0.031 \\ 
		 & &  50 & \bf 0.049 & \bf 0.049 &  0.040 \\ 
		 & & 100 & 0.051 & \bf 0.050 &  0.044 \\ 
		 \cline{2-6}
		 360 & G  & 20 & \bf 0.044 & 0.042 & 0.033 \\ 
		 & &  50 & \bf 0.050 & \bf 0.050 & 0.042 \\ 
		 & & 100 & \bf 0.050 & 0.049 & \bf 0.050 \\ 
		\cline{3-6}
		 & T & 20 & 0.046 & \bf 0.047 & 0.030 \\ 
		 & &  50 & \bf 0.050 & 0.048 &  0.045 \\ 
		 & & 100 & \bf 0.050 & 0.049 & 0.045  \\ 
		 \hline
	\end{tabular}
\end{table}





\begin{table}
	\centering
	\caption{Type I Error for doubly ranked Kruskal-Wallis tests by $S$, distribution of $Z_k$, and $n$ under AR(1) noise. Each value is based on 10000 simulated datasets. Values closest to nominal are bolded. G stands for Gaussian, T for $t_2$, Avg. for the average rank, and Suff. for the sufficient statistic. \label{t:kwe}}
	\begin{tabular}{lllcc}
		  \hline
		  \multirow{2}{*}{$S$} & \multirow{2}{*}{$Z_k$} & \multirow{2}{*}{$n$} &  \multicolumn{2}{c}{Doubly Ranked} \\
		 \cline{4-5}
		  &  &  & Avg. & Suff.  \\ 
		  \hline
		 40 & G  & 30 & 0.047 & \bf 0.050 \\ 
		 & &  75 & \bf 0.048 & 0.046  \\ 
		 & & 150 & 0.046 & \bf 0.047  \\ 
		\cline{3-5}
		 & T & 30 & 0.046 & \bf 0.048 \\ 
		 & &  75 & 0.051 & \bf  0.050 \\ 
		 & & 150 & \bf 0.048 & 0.046  \\ 
		 \cline{2-5}
		 120 & G  & 30 & 0.047 & \bf 0.048 \\ 
		 & &  75 & \bf 0.046 & \bf 0.046 \\ 
		 & & 150 & \bf 0.046 &  0.045 \\ 
		\cline{3-5}
		 & T & 30 & 0.045 & \bf 0.046 \\ 
		 & &  75 & \bf 0.048 &  0.047 \\ 
		 & & 150 & \bf 0.049 & \bf 0.049 \\ 
		 \cline{2-5}
		 360 & G  & 30 & \bf 0.049  & 0.048 \\ 
		 & &  75 & 0.045 & \bf 0.046  \\ 
		 & & 150 & \bf 0.047 & \bf 0.047  \\ 
		\cline{3-5}
		 & T & 30 &0.047 & \bf 0.048 \\ 
		 & &  75 & 0.046 & \bf 0.047  \\ 
		 & & 150 & \bf 0.048 & \bf 0.048 \\ 
		 \hline
	\end{tabular}
\end{table}

\begin{figure*}
	\centering
	\includegraphics[width = 2.1in, height = 2.1in]{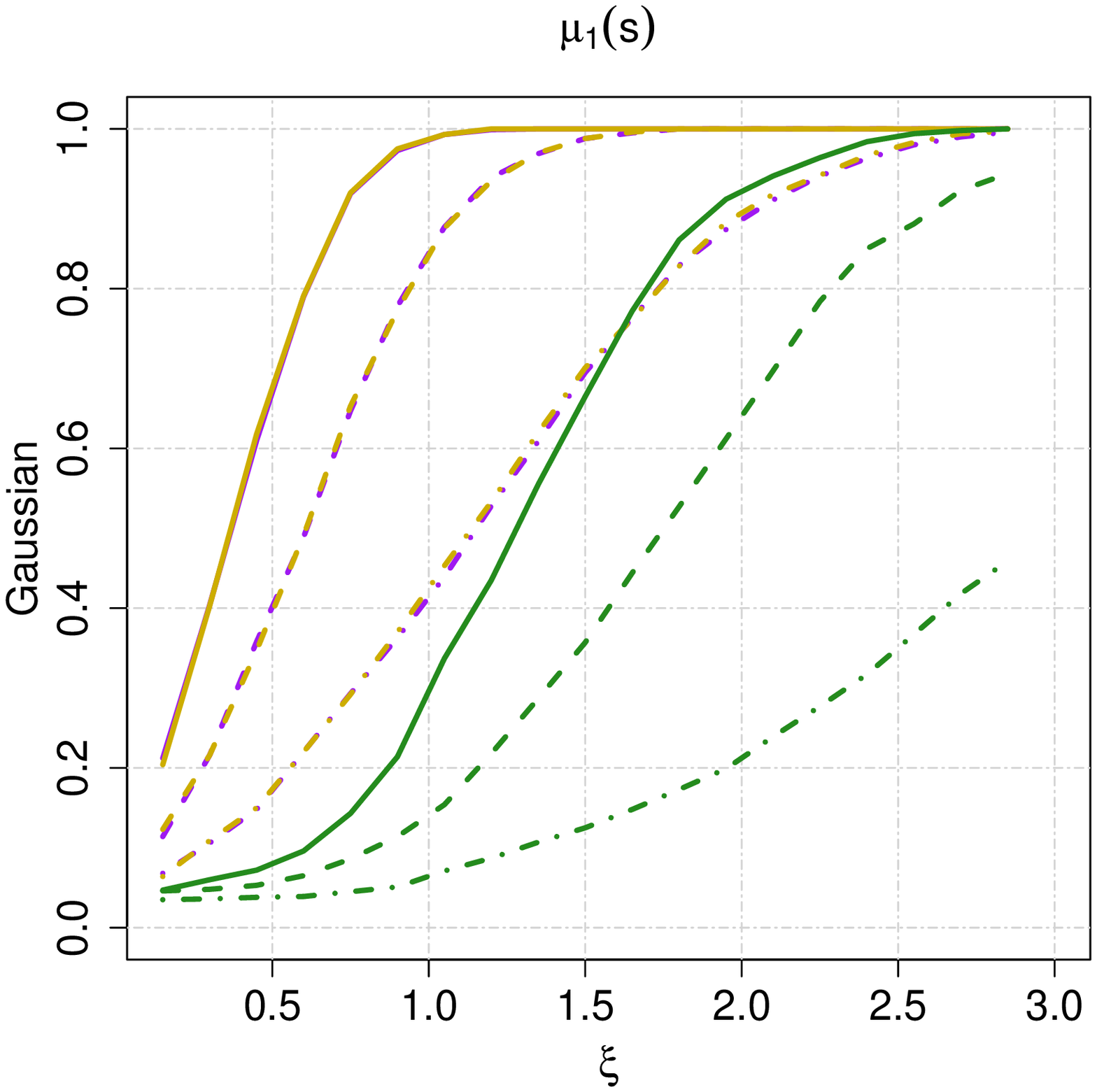}
	\includegraphics[width = 2.1in, height = 2.1in]{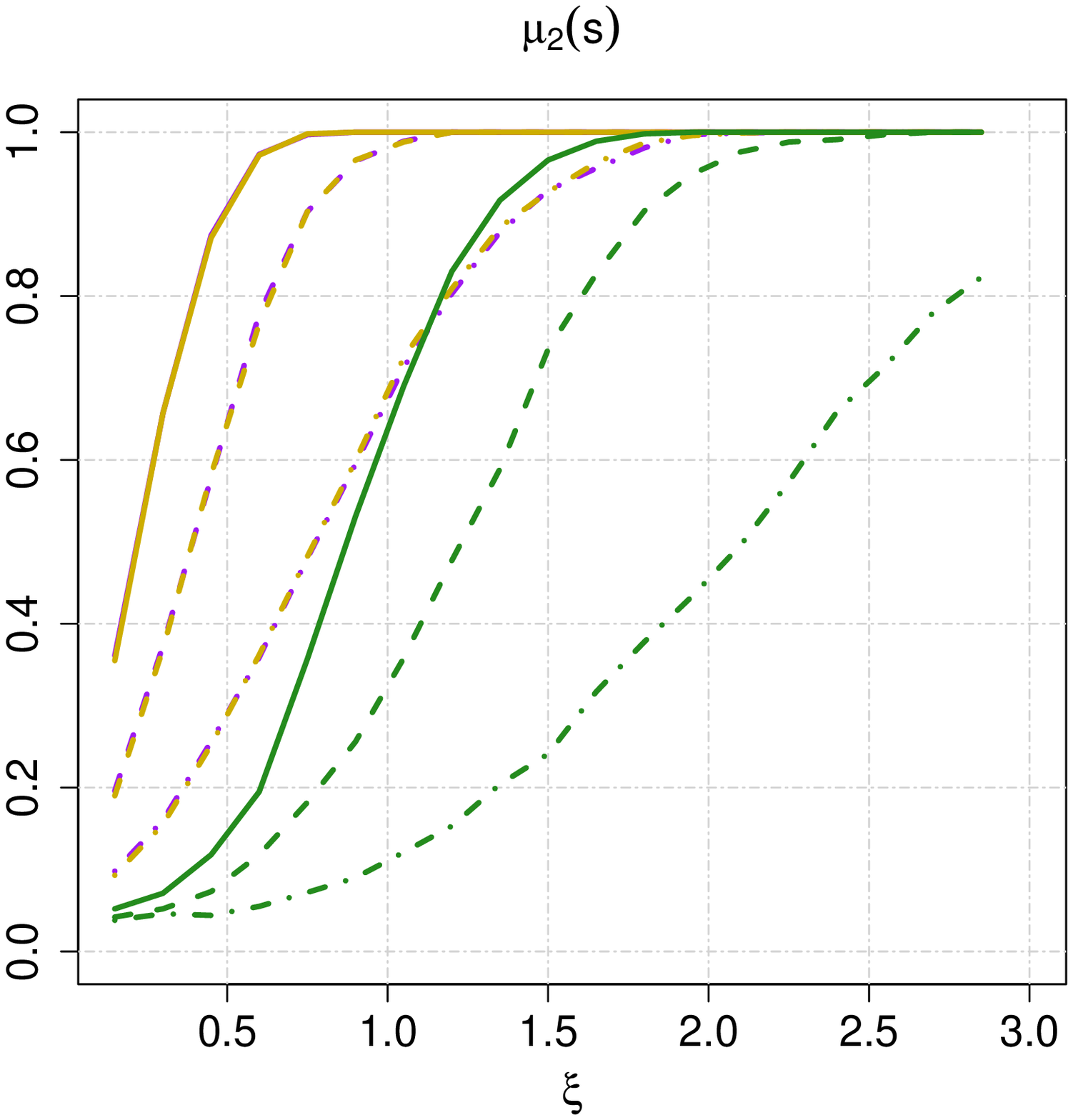}
	\includegraphics[width = 2.1in, height = 2.1in]{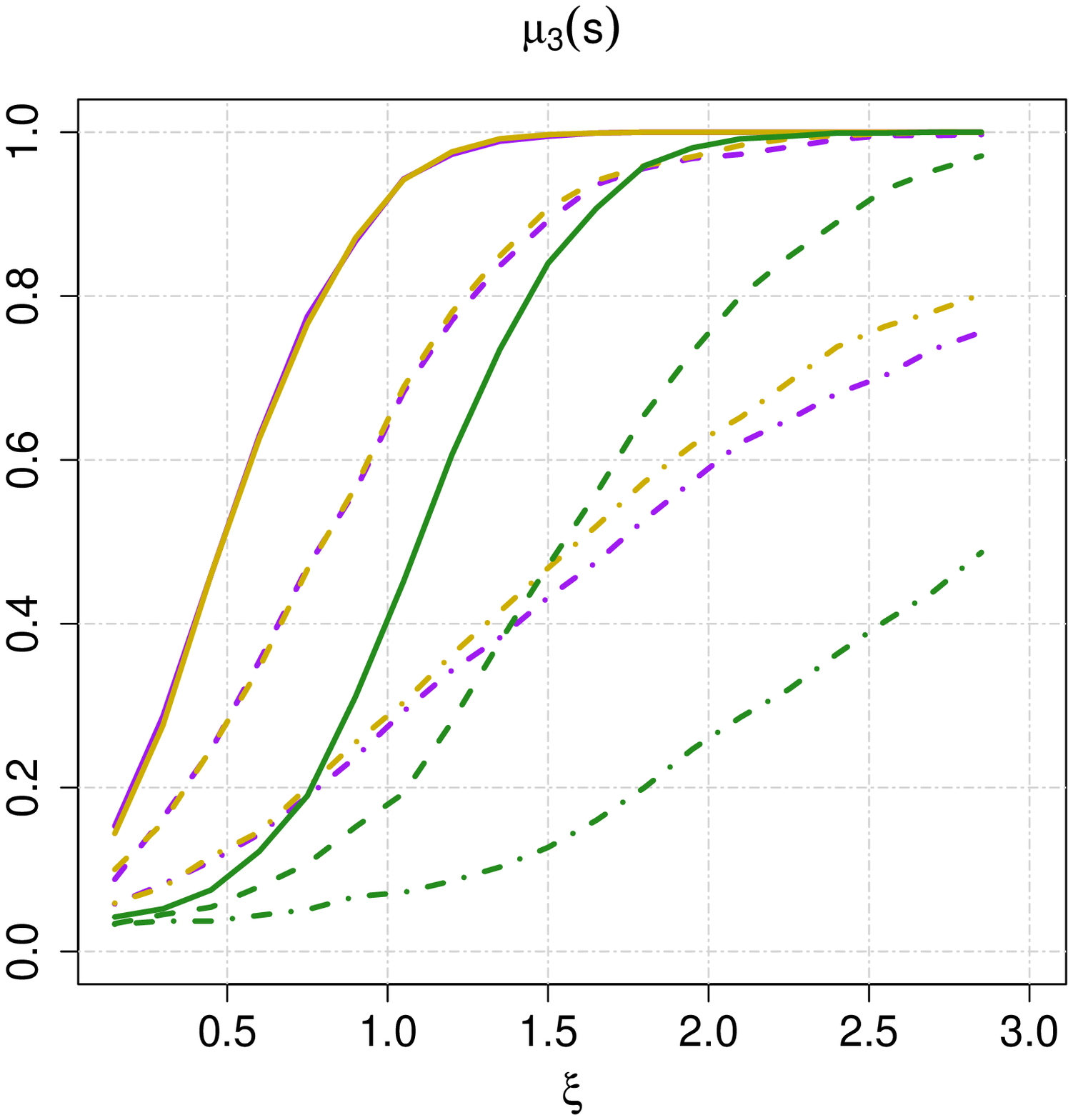}
	\includegraphics[width = 2.1in, height = 2.1in]{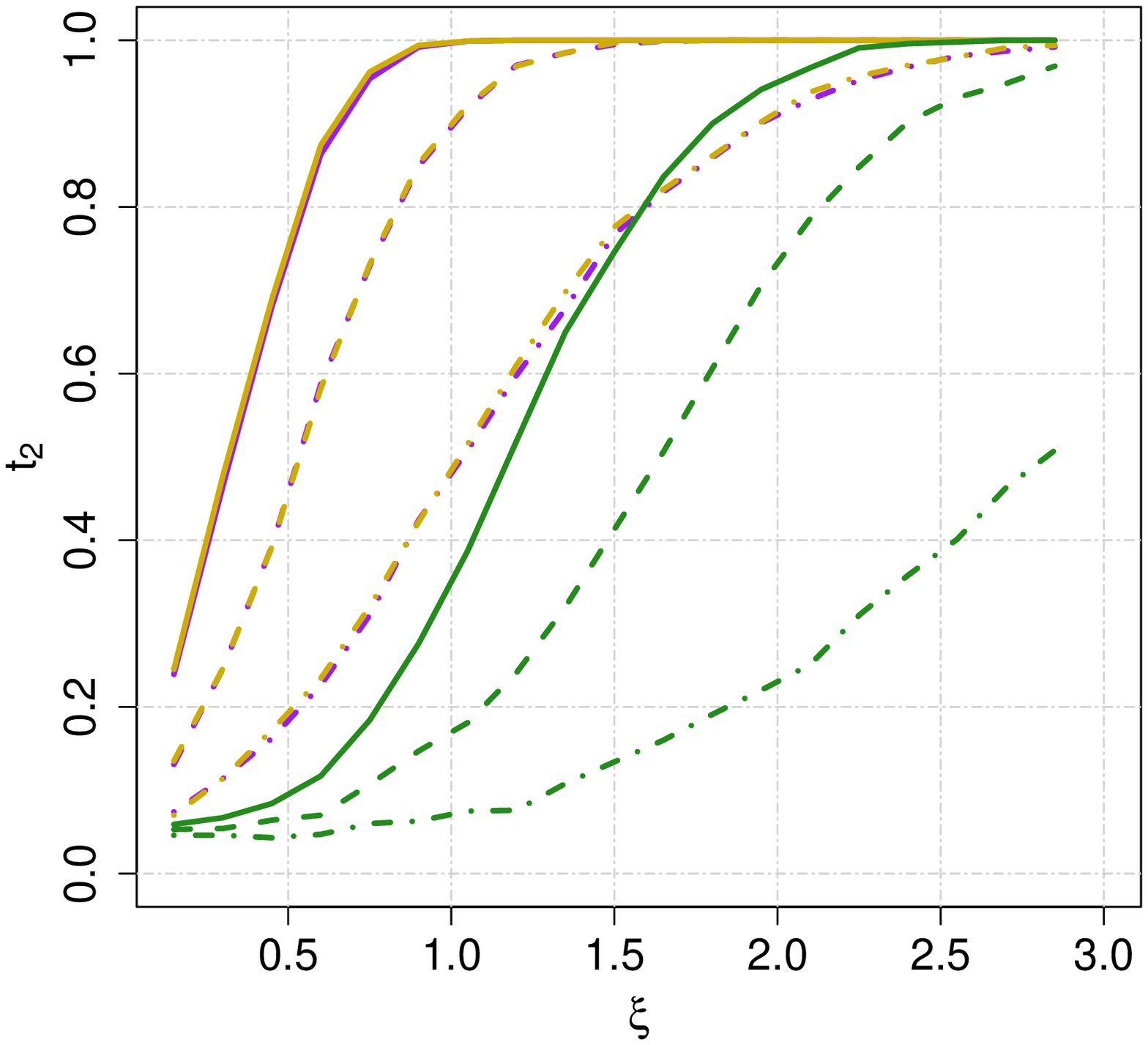}
	\includegraphics[width = 2.1in, height = 2.1in]{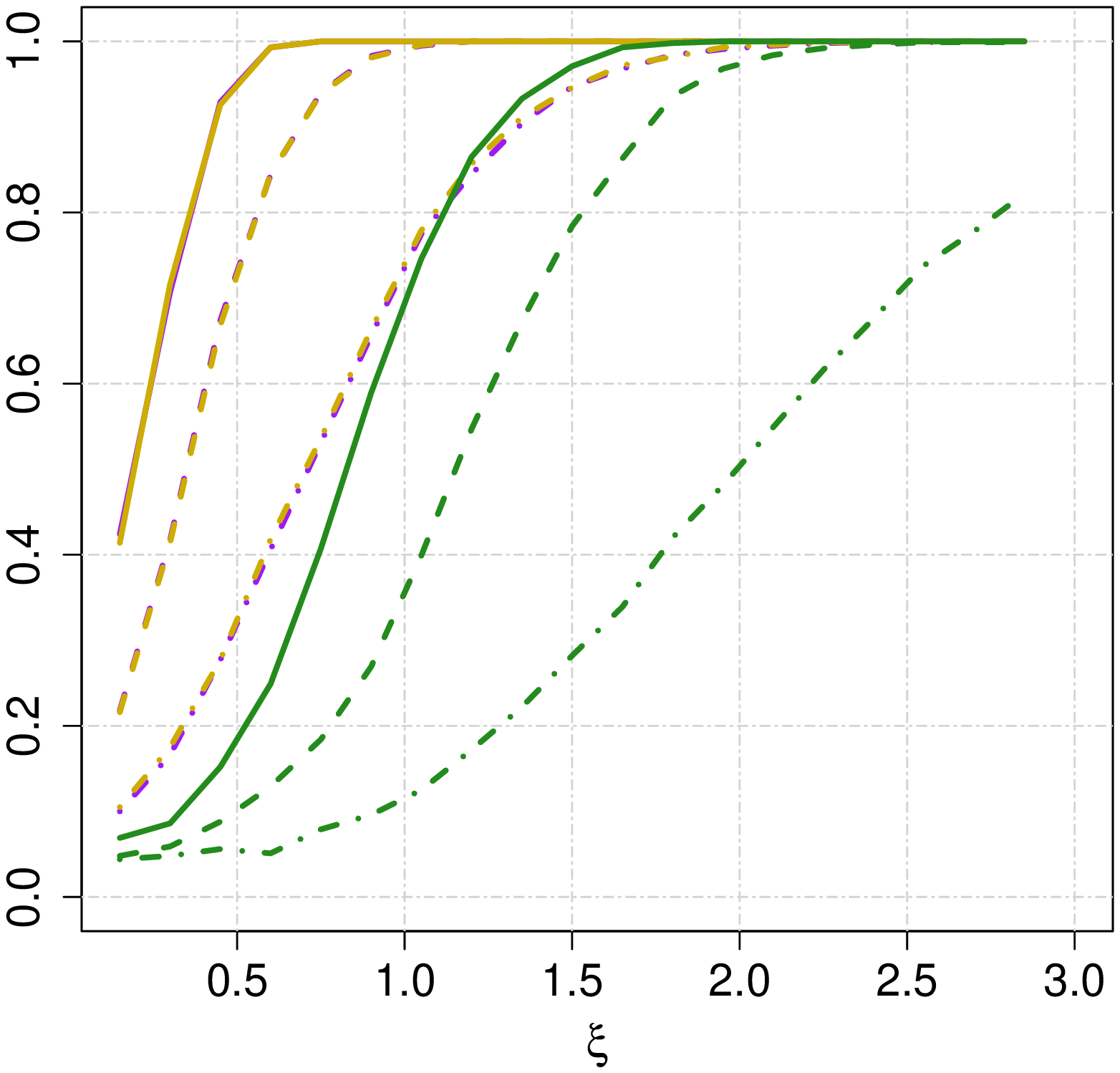}
	\includegraphics[width = 2.1in, height = 2.1in]{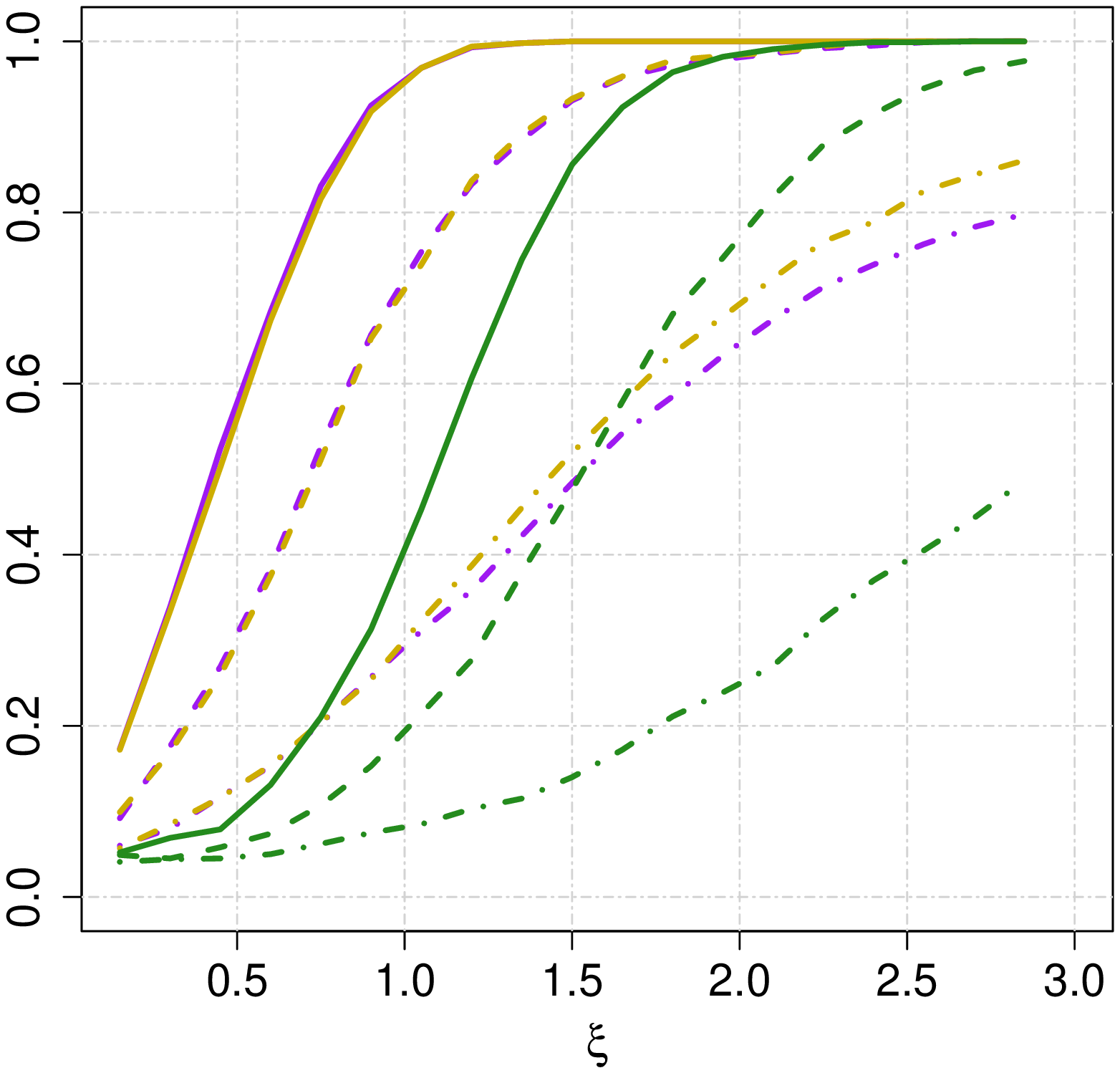}
	\caption{Power curves for functional Mann-Whitney-Wilcoxon tests under the sparsest sampling density, $S = 40$, with AR(1) noise. Rows index varying distributions of $Z_k$, columns index values of $\mu(s)$. Curves for the doubly ranked test in purple (sufficient statistic) and gold (average), curves for depth-based tests are in green. Solid curves are for when $n = 100$, dashed curves for when $n = 50$, and dotted-dashed curves for when $n = 20$.\label{f:pm40ar}}
\end{figure*}

\begin{figure*}
	\centering
	\includegraphics[width = 2.1in, height = 2.1in]{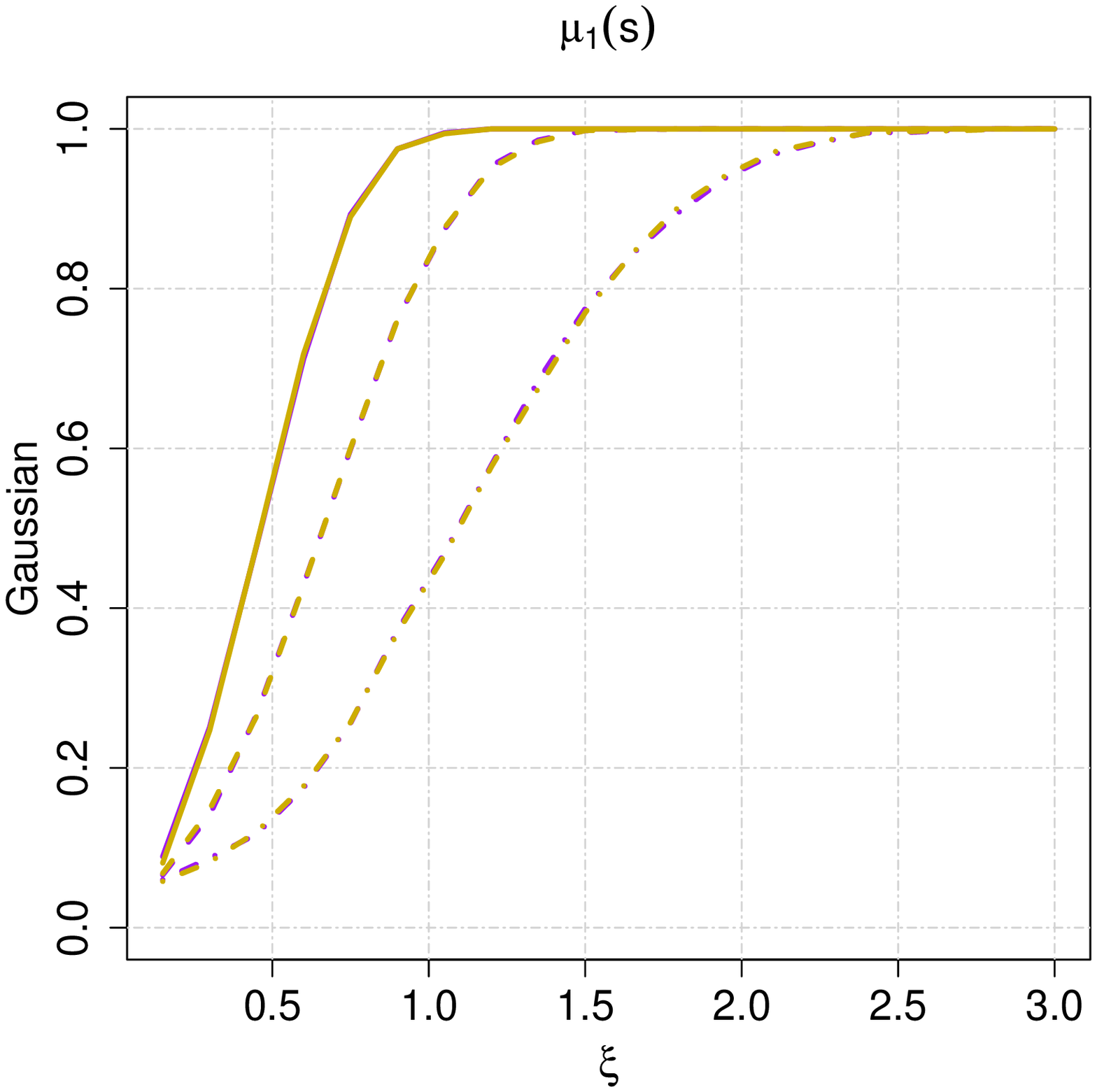}
	\includegraphics[width = 2.1in, height = 2.1in]{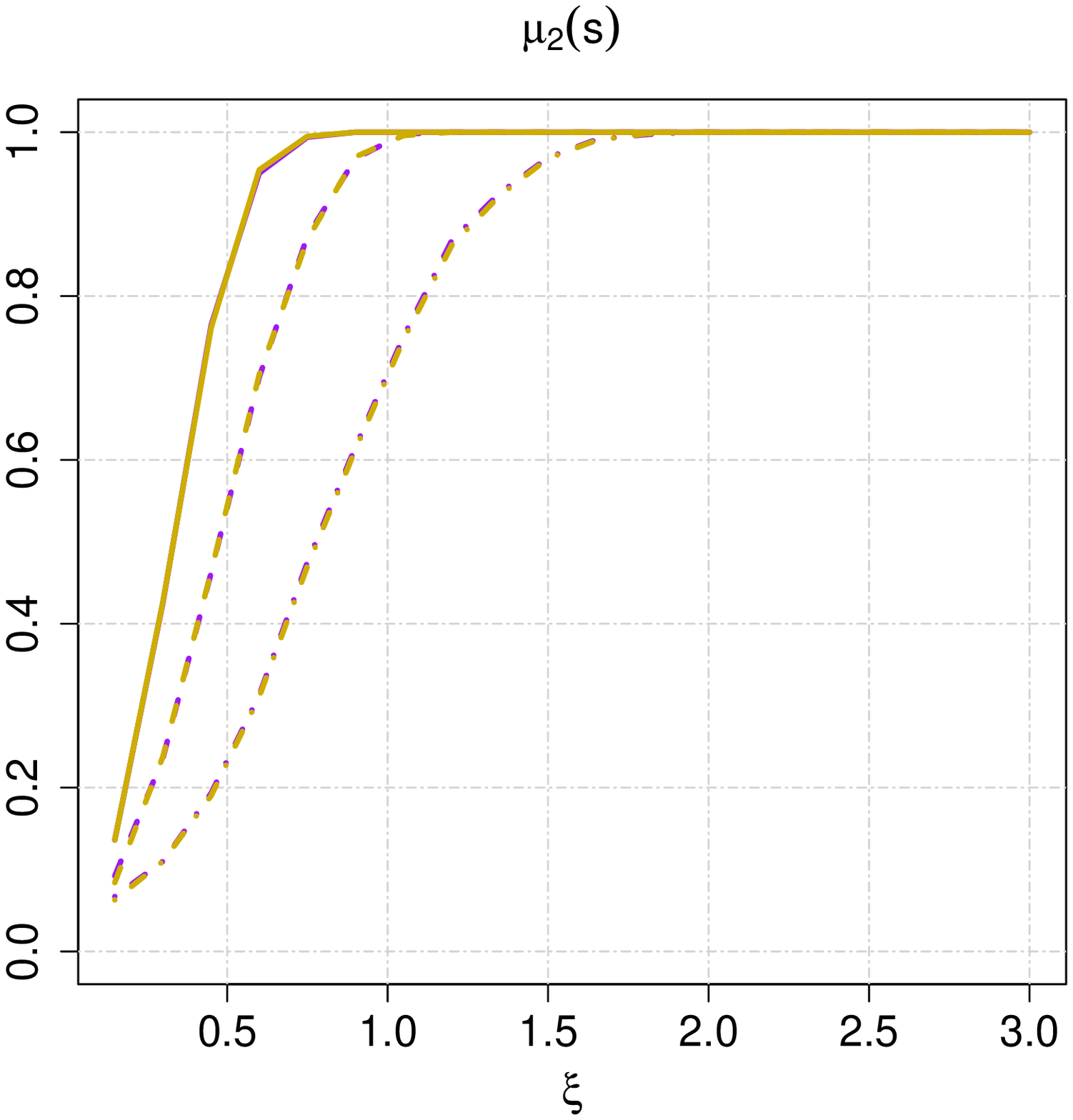}
	\includegraphics[width = 2.1in, height = 2.1in]{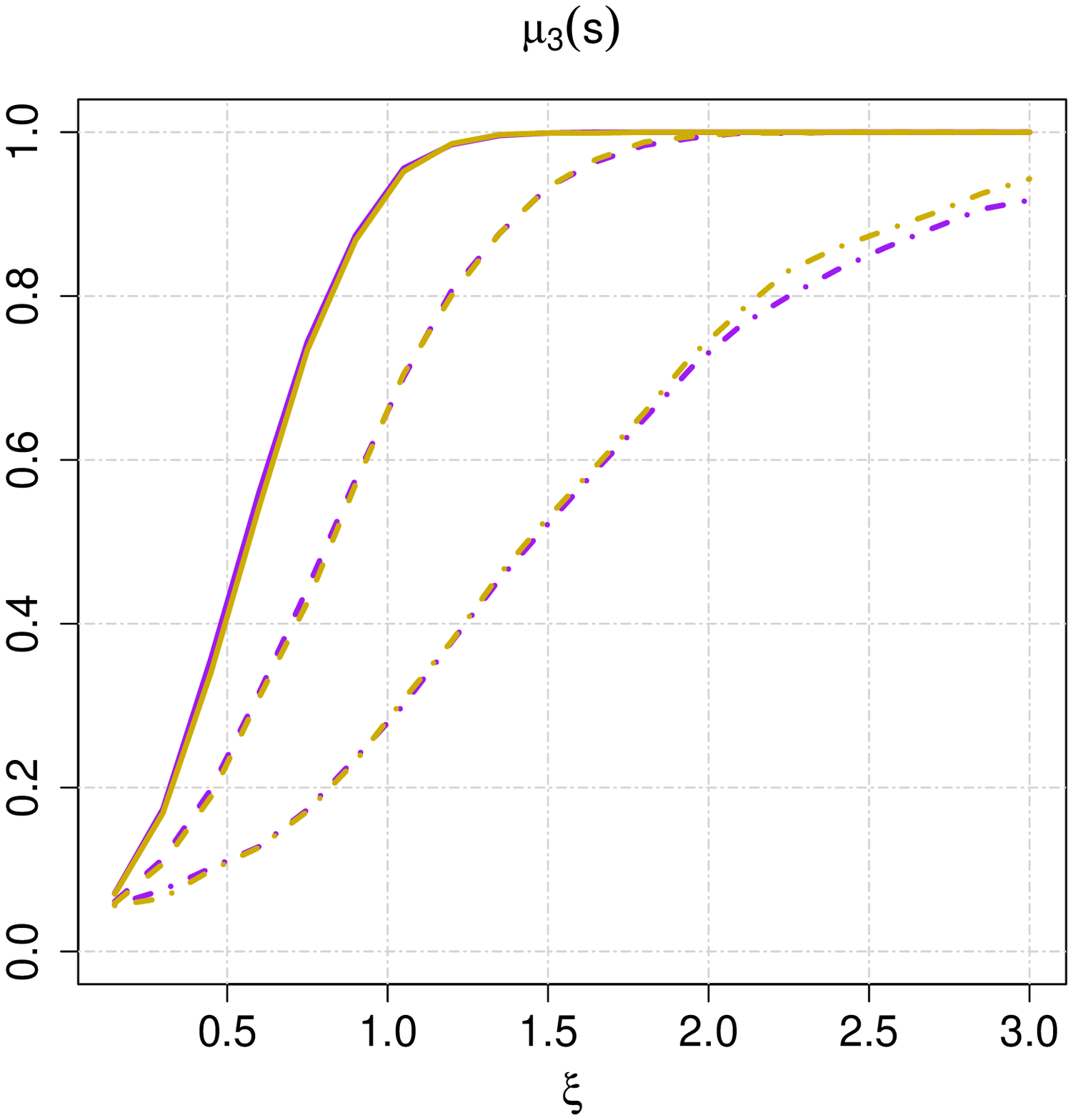}
	\includegraphics[width = 2.1in, height = 2.1in]{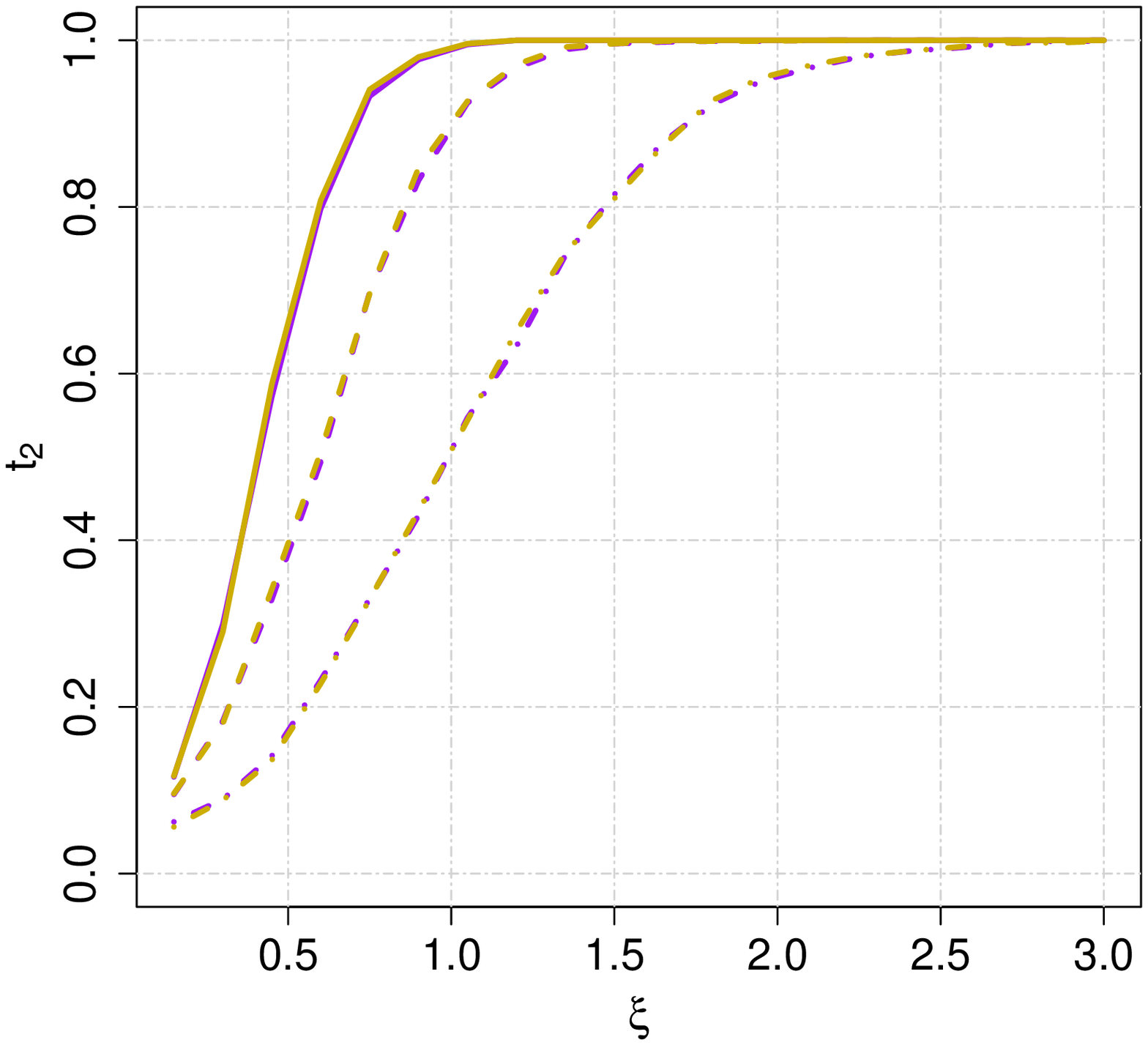}
	\includegraphics[width = 2.1in, height = 2.1in]{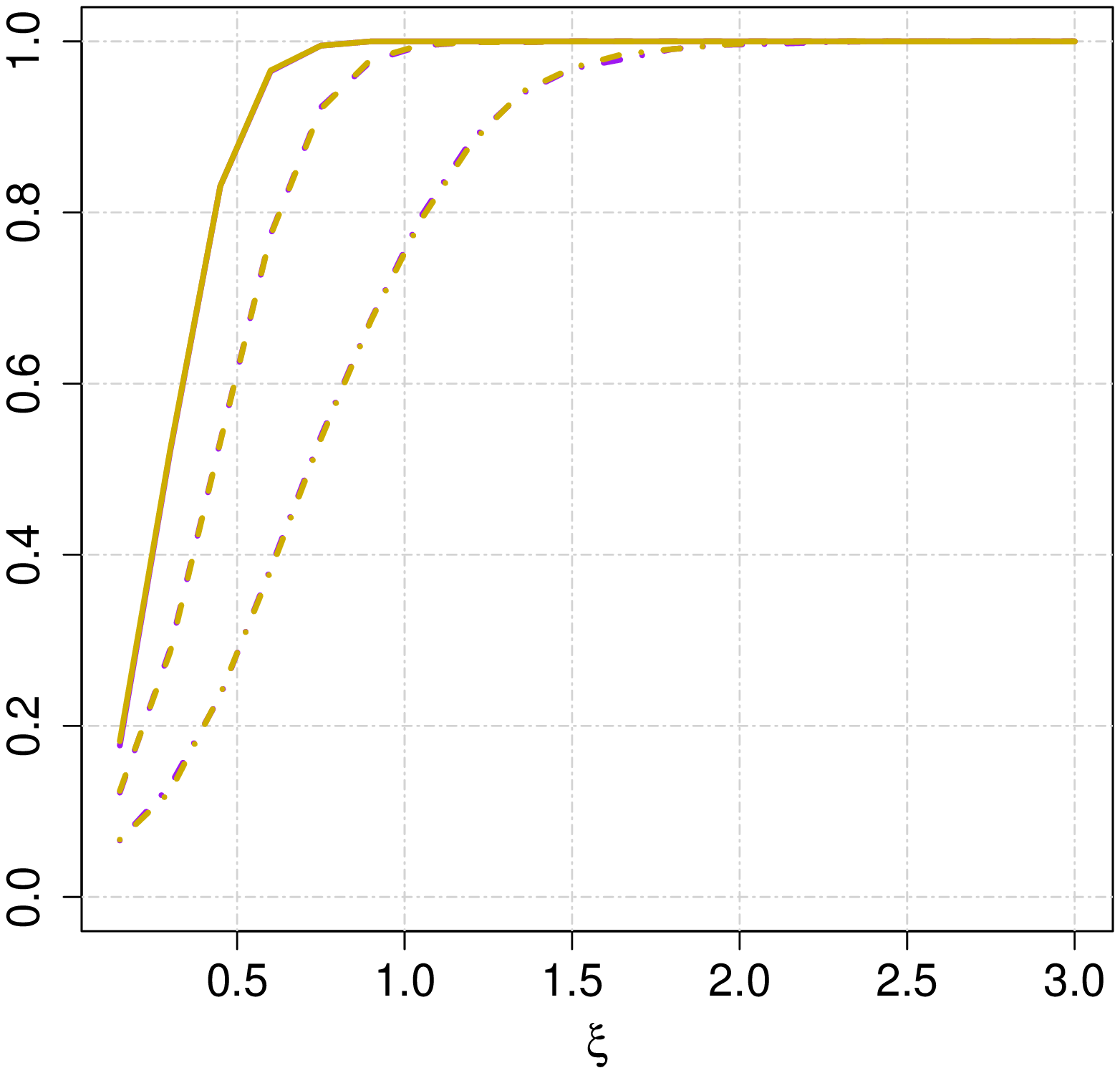}
	\includegraphics[width = 2.1in, height = 2.1in]{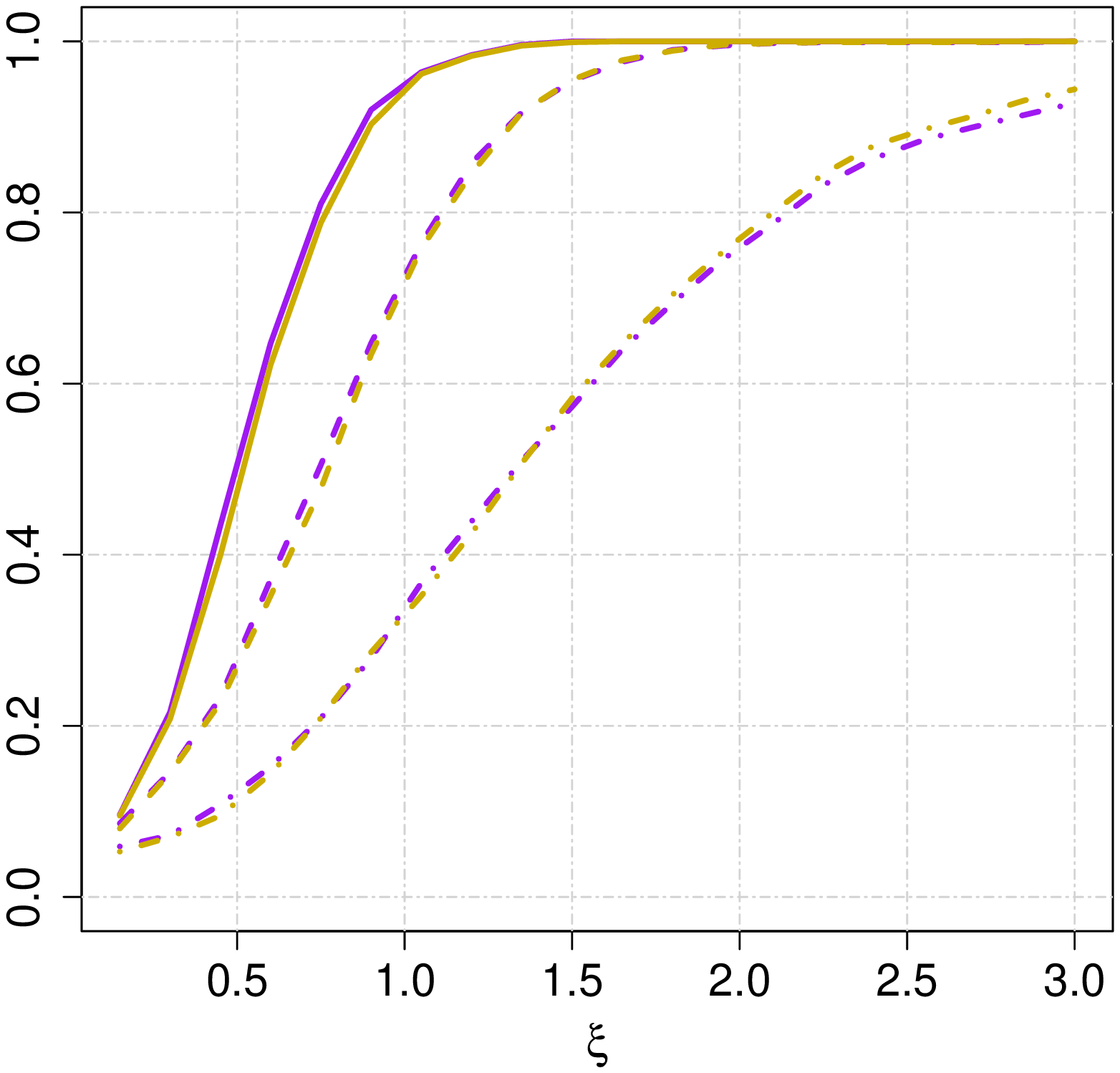}
	\caption{Power curves for functional Kruskal-Wallis tests under the sparsest sampling density, $S = 40$, with AR(1) noise. Rows index varying distributions of $Z_k$, columns index values of $\mu(s)$. Curves for the doubly ranked test in purple (sufficient statistic) and gold (average). Solid curves are for when $n = 100$, dashed curves for when $n = 50$, and dotted-dashed curves for when $n = 20$.\label{f:pk40ar}}
\end{figure*}




	Type I error estimates for the the doubly ranked MWW tests and the depth-based test are presented in Table~\ref{t:mwwe} while Table~\ref{t:kwe} contains type I error for the doubly ranked KW tests. Tables are broken down by $S$, distribution of $Z_k$ and total sample size, $n$. They contain only results for the AR(1) noise simulation. Each empirical type I error rate is based on 10000 simulated datasets. Testing was performed at the $\alpha = 0.05$ level. Bolded figures have the closest to nominal type I error rate. Regardless of $S$, $Z_k$, or $n$, all methods produce rates that are either below nominal or at most within Monte Carlo error of it. In the MWW case, a doubly ranked test is always closest to nominal (or tied for closest) and in only one instance is the depth-based tied for closest. Type I error tends to be slightly smaller under the smallest sample size but the doubly ranked tests never fall below 0.042 whereas the depth-based test has type I error consistently below 0.035 when $n = 20$. Comparing between the doubly ranked tests, the sufficient statistic tends to be closer to nominal overall, however the average rank performs similarly in a number of settings. The sufficient statistic's advantage is more noticeable in the KW case, although the type I error rates for the average rank are similar. Additional results for type I error under the other measurement error settings, no noise and white noise, are in Section 2 of Supplement A \citep{DRTSup2024A}.

Power curves are based on 500 simulated datasets for each combination of $\xi$, $S$, $n$, $\mu(s)$, and distribution of $Z_k$. Figure~\ref{f:pm40ar} presents the curves under AR(1) noise and the sparsest sampling density ($S = 40$) since increasing the sampling density only modestly impacts the power of the doubly ranked MWW tests. Solid curves denote the settings where $n = 100$, dashed were $n = 50$, and dotted-dashed for when $n = 20$. Curves for the doubly ranked MWW tests are in purple (sufficient statistic) and gold (average rank) while curves for the depth-based MWW test are in green. The rows of Figure~\ref{f:pm40ar} index the distribution of $Z_k$ and the columns index the varying functional forms of $\mu(s)$. Overall, power increases for all approaches as $n$ increases, no matter the distribution of $Z_k$ and the form of $\mu(s)$. Comparing the two frameworks and at a given $n$, the doubly ranked MWW tests have consistently higher power than the depth-based MWW test. Choice of summary does not substantially impact power as the curves are nearly identical between the sufficient statistic and the average rank. Similar power curves to those in Figure~\ref{f:pm40ar} when $S = 120$ or 360 are in Section 2 of Supplement A \citep{DRTSup2024A} along with curves for the remaining measurement error simulations. All methods tend to have slightly more power in the presence of less noise. Power also increases slightly as the density of the sampling grid increases.

Figure~\ref{f:pk40ar} contains the power curves for the doubly ranked KW tests when $S = 40$ under AR(1) noise. As with the MWW tests, the curves are based on 500 simulated datasets per combination of $\xi$, $n$, $\mu(s)$, and distribution of $Z_k$. Solid curves now denote the settings where $n = 150$, dashed where $n = 75$, and dotted-dashed for when $n = 30$. The sufficient statistic curves are in purple while the average rank curves are in gold. Overall, power increases as $n$ increases across $Z_k$ and forms of $\mu(s)$. The form of $\mu(s)$ does have a slight impact on power when $n$ is small but the distribution of $Z_k$ does not appear to have much of an impact. The sufficient statistic performs slightly better under the $\mu_3(s)$ setting but overall, the power curves are largely indistinguishable between summary statistics. Additional power curves when $S = 120$ and 360 and for less noisy simulations are in Section 2 of Supplement A \citep{DRTSup2024A}. As with the MWW case, less noise tends to lead to slightly increased power as does increasing the density of the grid.

\section{Data Illustrations}
\label{s:app}

To illustrate both doubly ranked MWW and KW tests, we examine functional data from three different substantive fields with varying numbers of observations and differing lengths of functions. Two of the datasets are available in \texttt{R} packages and the third is available in Supplement B \citep{DRTSup2024B}. Each dataset has at least one factor of interest that may differentiate the outcome curves. While we will draw light conclusions based on the findings, the purpose of these illustrations is to demonstrate the use of doubly ranked tests and not to elaborate on the scientific findings (or lack-thereof) that result from our analysis. Results, using the sufficient statistic, from all three data sources for various combinations of outcomes and factors are in Table~\ref{t:data}. Section 3 of Supplement A contains a similar table using the average rank \citep{DRTSup2024A}. All outcomes were first pre-processed using FACE, retaining 99\% of the variation within the curves. Code to run each example is available in Supplement B \citep{DRTSup2024B}.

\begin{table}
\centering
\caption{ Test statistics and $p$-values for doubly ranked tests using the sufficient statistic applied to the three data illustrations: resin viscosity, Canadian weather, and COVID mobility. Temp. is short for temperature, Cur. for curing, Rot. for rotational, and Precip. for precipitation. Doubly ranked KW tests were performed for both Canadian weather outcomes and the COVID mobility factor MD, VA, \& WV. The remaining tests were doubly ranked MWW tests. \label{t:data}}
\vspace{5pt}
\begin{tabular}{lllcc}
  \hline
\multirow{2}{*}{Data Set} & \multirow{2}{*}{Outcome} & \multirow{2}{*}{Factor} &  Test & \multirow{2}{*}{$p$-value}  \\
 &   &  & Statistic &  \\ 
  \hline
Resin & 
Viscosity  &  Resin Temp. & 194 & $<0.001$   \\
 & & Cur. Agent Temp. & 440 & 0.337   \\
 & & Tool Temp. & 207 & $<0.001$   \\
& & Rot. Speed & 419.5 & 0.217  \\
& & Mass Flow  & 507.5 & 0.957  \\
\cline{2-5}
Weather & Temp.  &  Region & 21.44 & $<0.001$   \\
\cline{3-5}
& Precip.  &  Region & 22.46 & $<0.001$  \\ 
\cline{2-5}
COVID & Requests & CO \& UT & 132 & $<0.001$  \\
 & & IA \& MN & 1871 & $<0.001$ \\
 & & MD, VA, \& WV & 2.214 & 0.331   \\
    \hline
\end{tabular}
\end{table}


\subsection{Resin Viscosity}

The resin viscosity data comes from an experimental setting conducted at the Technical University of Munich's Institute for Carbon Composites and is freely available in the \texttt{R} package \texttt{FDBoost} \citep{FDboost2020}. The data set contains measurements of the viscosity of resin over the course of 838 seconds with the goal of assessing factors that influence the curing process in a resin mold. The grid is not equally spaced due to technical reasons whereby viscosity can be measured every two seconds at first but, due to hardening, can only be measured every ten seconds after the 129th second. In total, 64 different molds were poured under five different experimental conditions: temperature of resin, temperature of the curing agent, temperature of the tools, rotational speed, and mass flow. Each condition is a binary factor, with a ``low'' and ``high'' levels.

	\begin{figure}
		\centering
		\includegraphics[width = 2.1in, height = 2.1in]{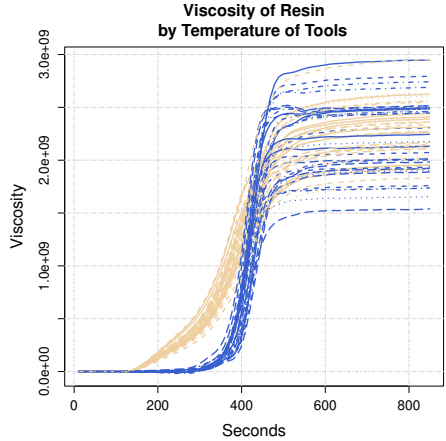}
		\caption{Blue curves were measured under the ``low'' temperature condition while tan curves were measured under the ``high'' temperature condition.\label{f:resin}}
	\end{figure}

From Table~\ref{t:data}, we observe that in two of five factors the ranks tend to significantly differ between groups. Specifically, the resin temperature ($W = 194$, $p <0.001$) and tool temperature ($W = 207$, $p < 0.001$) have $p$-values below nominal. This suggests that the ranks of the curves that describe the viscosity of the resin differ significantly between the groups defined by these factors. To determine which group is higher or lower and when, we turn to a graphical assessment. When considering the tool temperature, Figure~\ref{f:resin} displays clear clustering in the curves with low temperature tools appearing to induce slower curing than high temperature tools. This separation lessens around 400 seconds where the sets of curves begin to overlap. The analysis by temperature of resin captures greater separation after 400 seconds which can be seen in the additional figures in Section 3 of Supplement A \citep{DRTSup2024A}. The remaining factors did not significantly distinguish the curves as can also be seen in the graphs in Section 3 of Supplement A \citep{DRTSup2024A}.



\subsection{Canadian Weather}

The Canadian weather data is a classic dataset for illustrating functional data analysis techniques having appeared in texts by \cite{Ramsay2005} and \cite{Ramsay2009} as well as in many articles that are too numerous to list here. It contains daily average temperature (C) and precipitation (mm) measurements, averaged over the years 1960 to 1994, taken at 35 different sites across Canada. The factor of interest is the region from which the measurement was taken: Arctic, Atlantic, Continental, or Pacific. It is available in the \texttt{R} package \texttt{fda} \citep{fda2022}. For the purpose of this illustration, the goal of the analysis is to see if the doubly ranked KW test can detect differences in temperature and precipitation curves by region as might be expected given their geographic differences.

	\begin{figure}
		\centering
		\includegraphics[width = 2.1in, height = 2.1in]{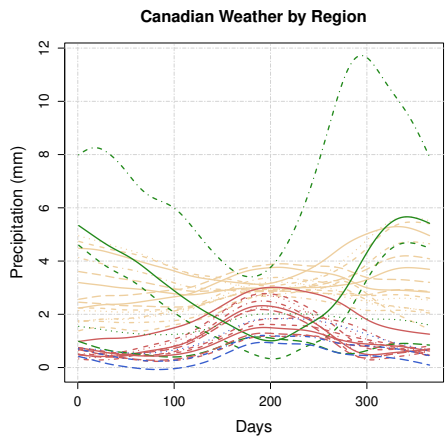}
		\caption{Blue curves denote the stations in the Arctic region, tan curves indicate stations the Atlantic region, red curves for stations in the Continental region, and green curves are the Pacific region.\label{f:precip}}
	\end{figure}

Table~\ref{t:data} contains the results of the analysis for both outcomes. We see that the doubly ranked KW returns significant tests at the nominal level for both temperature curves ($X^2 = 21.44$, $p <0.001$) and precipitation curves ($X^2 = 22.46$, $p < 0.001$) differing by region. Formally, we would say the location parameters of the distributions of the summarized ranks for temperature and precipitation over time differ by region. To observe the nature of the difference by group, we examine Figures~\ref{f:cwtemp} and~\ref{f:precip}. We see that the temperature curves from the Arctic region are the coldest while those from the Pacific region tend to be warmest. The precipitation curves show inverted patterns based on region with the Pacific region tending to see more precipitation during the winter months while the Continental and Arctic regions experience more during the summer months.

\subsection{COVID Mobility}

The COVID mobility data was extracted from publicly available data provided by Apple Inc. regarding daily mapping requests for driving, walking, and transit directions during the onset of the COVID-19 pandemic and for some time thereafter. A cached version of the data is available from \cite{Gassen2022}. The data itself represents the percent change in requests for directions from a baseline date of January 13, 2020. Requests are categorized by whether they were for driving, walking, or transit directions. In our analysis, we focus on the percent change in daily requests for driving directions. For the United States, the data is available at the county level which is our unit of observation.

Rather than aligning the data by calendar time, we align it by the first date of a major COVID policy implementation at the state level, e.g. shelter-in-place or safer-at-home orders. Policy dates were obtained from the publicly available COVID AMP data \citep{Katz2023}. The timeline is centered at the date of the first COVID-policy implementation in each state and we take 30 days prior and 30 days post. The main objective of this analysis is to compare different states in similar regions to each other, focusing on whether or not there is a difference post-policy implementation. Specifically, we compare county level percent changes in driving requests between counties in Iowa (IA) and Minnesota (MN), Colorado (CO) and Utah (UT), and finally between Maryland (MD), Virginia (VA), and West Virginia (WV). Data for this illustration is available in Supplement B \citep{DRTSup2024B}.

	\begin{figure}
		\centering
		\includegraphics[width = 2.1in, height = 2.1in]{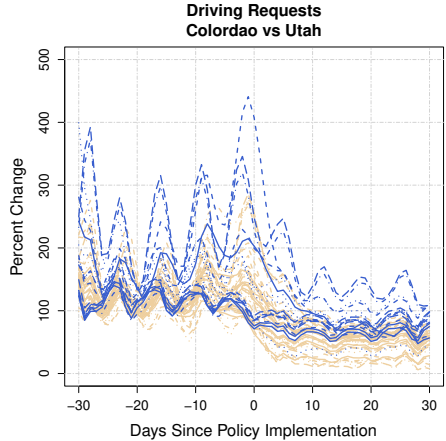}
		\caption{Blue curves denote counties in Utah while tan curves denote counties in Colorado.\label{f:cout}}
	\end{figure}

When comparing the percent change in county-level requests for driving directions between Colorado and Utah, we find a significant difference between the groups of curves by state ($W = 132$, $p = 0.001$). Figure~\ref{f:cout} shows that the county-level curves have substantial overlap pre-initial policy implementation. After implementation, however, Colorado counties tend to experience a lower percent change in requests. Graphs for the remaining comparisons are in Section 3 of Supplement A \citep{DRTSup2024A}. We observe a similar trend when comparing Minnesota and Iowa with Minnesota counties tending to experience lower percent change in driving requests which continues until three to four weeks post-implementation. This is reflected in the doubly ranked MWW test which suggests a significant difference in the location parameters of distributions of the summarized ranks between Minnesota and Iowa ($W = 1871$, $p < 0.001$). In contrast to the first two comparisons, we lack evidence to suggest a difference in the ranks of the curves between Maryland, Virginia, and West Virginia ($X^2 = 2.214$, $p = 0.331$).


	\section{Discussion}
	\label{s:disc}
	
	In this manuscript, we investigate two novel nonparametric tests for comparing groups of functional curves: the doubly ranked MWW test and the doubly ranked KW test. Both tests rely on the same doubly ranked framework which ranks first by time point before summarizing and then reranking for analysis. We consider two ways to construct a summary statistic of the ranks over time, first deriving a sufficient statistic under $H_0$ and then examining the properties of the average rank under $H_0$. The summarizing step is a data reduction step, thus we prefer the sufficient statistic given the sufficiency principle. We also observe some advantages in simulation to using the sufficient statistic. However, the average rank relies on central tendency and performs similarly. Thus, this summary could be of use in clinical settings where an average rank would be more digestible to non-specialists.
	While other approaches reduce dimensionality using a summary score for the data, only our sufficient statistic is constructed under $H_0$. Doubly ranked tests are distribution free since they ultimately rely only on the ranks of the sufficient statistics or average ranks. We demonstrate empirically that both the doubly ranked MWW and KW tests perform well in terms of type I error and power for different distributions under varying sample sizes, curve types, sampling densities, and measurement error models.
	
	The test procedures we propose here are global tests of differences between sets of curves, as are the other methods by \cite{Hall2007,Lopez2009,Lopez2010,Chak2015,Berrett2021} and \cite{Melendez2021}. With any global test, one must be careful to examine the underlying data for where the differences lie. These tests can instruct us that there is difference but cannot be used to identify exactly where the difference is nor can they tell us the magnitude of difference. For that, we would need to model the group location curves themselves and rely on point-wise testing, ideally adjusted for multiple comparisons. See, for example, some of the work on functional regression reviewed in \cite{Morris2015} and \cite{Greven2017}.
	
	Point-wise testing is not goal of this work, however our method can be useful as a global test in a variety of data contexts. In our data illustrations, we demonstrate the use of the doubly ranked MWW and KW tests to study problems in material science, climatology, and public health policy. 
	These tests can be performed as a main analysis or a first pass analysis in conjunction with other modeling steps. Our tests are also easy to implement and easy for practitioners to interpret---in particular, for those who are already familiar with univariate MWW and KW tests. Doubly ranked testing is not limited to functional data in the form of FACE-preprocessed curves. The assumptions of the tests apply broadly to other preprocessing techniques as well as potentially other functional data types, i.e. surfaces. However, the present work only considers data in the form of curves.
%
	
	Our method assumes the sampling grid for each curve is the same for each subject. Subjects may, however, have mistimed measurements and therefore asynchronous grids. In the current context, an FPCA that accounts for asynchronous grids can be employed should such data arise. For example, the FPCA by smoothed covariance can be used for sparse longitudinal data, see \cite{Yao2005} or \cite{Di2009} among others. Asynchronous functional data is a more general case of sparse longitudinal data. Thus, such an approach could be used for sparsely sampled profiles, although this manuscript only considers the complete curve case.

\begin{appendix}

\section{Proof of Claim~\ref{cl:ef}}\label{appx:ef}


\begin{proof}
Consider the mass function in Equation~(\ref{eq:pmf}) and rewrite it as
\begin{align*}
	&P\bigg[ Z_{(r)} = z \bigg] \approx \frac{\Gamma(n+1)}{\Gamma(r)\Gamma(n - r + 1)} \frac{1}{n}  \\
	&\hspace{15pt} \times  \left(\frac{z}{n} - \frac{1}{2n}\right)^{r-1} \left(1-\frac{z}{n} + \frac{1}{2n}\right)^{n-r}\\
	&= \frac{\Gamma(n+1)}{\Gamma(r)\Gamma(n - r + 1)} \frac{1}{n}  \\
	&\hspace{15pt} \times  \exp\bigg[ n\log\left(1-\frac{z}{n} + \frac{1}{2n}\right) - \log \left(\frac{z}{n} - \frac{1}{2n}\right) \bigg] \\
	&\hspace{15pt} \times \exp\bigg[ r\log \left\{\left(\frac{z}{n} - \frac{1}{2n}\right) \bigg/ \left(1-\frac{z}{n} + \frac{1}{2n}\right) \right\}\bigg]\\
\end{align*}
Noting that $n$ is fixed, define the following functions: 
\begin{align*}
	c(r) &= \frac{1}{n}\frac{\Gamma(n+1)}{\Gamma(r)\Gamma(n - r + 1)},\ w(r) = r,\\
	h(z) &= \exp\bigg[ n\log\left(1-\frac{z}{n} + \frac{1}{2n}\right) - \log \left(\frac{z}{n} - \frac{1}{2n}\right) \bigg],\\
	\text{and}\\
	 t(z) &= \log \left\{\left(\frac{z}{n} - \frac{1}{2n}\right) \bigg/ \left(1-\frac{z}{n} + \frac{1}{2n}\right) \right\}.
\end{align*} 
The mass function can then be written as
\begin{align*}
	P\bigg[ Z_{(r)} = z \bigg]  &= h(z)c(r)\exp\left[ w(r)t(z) \right],
\end{align*}
which is the form of an exponential family. 
\end{proof}

\section{Proof of Claim~\ref{cl:ssm}}\label{appx:ev}


\begin{proof}
First note the mass function of a discrete uniform random variable over $\{1, n\}$ is $P(Z = z) = \frac{1}{n}\ \forall\ z \in \{1, n\}$, $n \in \mathbb{Z}^+$. The expectation of $t(z)$ is
\begin{align*}
	E\left[ t(z) \right] &= E\left[ \log \left\{\left(\frac{z}{n} - \frac{1}{2n}\right) \bigg/ \left(1-\frac{z}{n} + \frac{1}{2n}\right) \right\} \right] \\
	&= \sum_{z = 1}^n \log \left\{\left(\frac{z}{n} - \frac{1}{2n}\right) \bigg/ \left(1-\frac{z}{n} + \frac{1}{2n}\right) \right\} \frac{1}{n}\\
	&= \frac{1}{n} \log \left\{ \prod_{z = 1}^n \left(\frac{z}{n} - \frac{1}{2n}\right) \bigg/ \left(1-\frac{z}{n} + \frac{1}{2n}\right) \right\} 
\end{align*}
Expanding out the first few terms and last few terms of the product, we note that it is a telescoping product:
\begin{align*}
	&\prod_{z = 1}^n \left(\frac{z}{n} - \frac{1}{2n}\right) \bigg/ \left(1-\frac{z}{n} + \frac{1}{2n}\right) = \\
	&\hspace{10pt} \frac{\frac{1}{n} - \frac{1}{2n}}{1-\frac{1}{n} + \frac{1}{2n}} \times \frac{\frac{2}{n} - \frac{1}{2n}}{1-\frac{2}{n} + \frac{1}{2n}} \times \frac{\frac{3}{n} - \frac{1}{2n}}{1-\frac{3}{n} + \frac{1}{2n}} \times \cdots\\
	&\hspace{10pt}\cdots \times \frac{\frac{n-2}{n} - \frac{1}{2n}}{1-\frac{n-2}{n} + \frac{1}{2n}} \times \frac{\frac{n-1}{n} - \frac{1}{2n}}{1-\frac{n-1}{n} + \frac{1}{2n}} \times \frac{\frac{n}{n} - \frac{1}{2n}}{1-\frac{n}{n} + \frac{1}{2n}}\\
	&= \frac{\frac{1}{2n}}{1-\frac{1}{2n}} \times \frac{\frac{1}{n} + \frac{1}{n} - \frac{1}{2n}}{1-\frac{1}{n} -\frac{1}{n} + \frac{1}{2n}} \times \frac{\frac{1}{n} + \frac{2}{n} - \frac{1}{2n}}{1-\frac{1}{n} - \frac{2}{n} + \frac{1}{2n}} \times \cdots\\
	&\hspace{10pt}\cdots \times \frac{1 - \frac{2}{n} - \frac{1}{2n}}{1- 1 + \frac{2}{n} + \frac{1}{2n}} \times \frac{1 -\frac{1}{n} - \frac{1}{2n}}{\frac{1}{n} + \frac{1}{2n}} \times \frac{1 - \frac{1}{2n}}{\frac{1}{2n}}\\
	&= \frac{\frac{1}{2n}}{1-\frac{1}{2n}} \times \frac{\frac{1}{n} + \frac{1}{2n}}{1-\frac{1}{n} -  \frac{1}{2n}} \times \frac{\frac{1}{n} + \frac{2}{n} - \frac{1}{2n}}{1 - \frac{2}{n} - \frac{1}{2n}} \times \cdots\\
	&\hspace{10pt}\cdots \times \frac{1 - \frac{2}{n} - \frac{1}{2n}}{\frac{1}{n} + \frac{2}{n} - \frac{1}{2n}} \times \frac{1 -\frac{1}{n} - \frac{1}{2n}}{\frac{1}{n} + \frac{1}{2n}} \times \frac{1 - \frac{1}{2n}}{\frac{1}{2n}}
\end{align*}
Examining the last two lines, the first fraction is the reciprocal of the last fraction, the second fraction is the reciprocal of the second to last fraction, the third fraction is the reciprocal of the third to last fraction, and so on. 

When $n$ is even, the first $\frac{n}{2}$ fractions will cancel with the last $\frac{n}{2}$ fractions. When $n$ is odd, the first $\left\lfloor\frac{n}{2} - 1\right\rfloor$ fractions will cancel with the last $\left\lfloor\frac{n}{2} - 1\right\rfloor$ fractions. The remaining fraction will be the case where $z = \frac{n+1}{2}$. When $z = \frac{n+1}{2}$, the ratio will be 1. Thus, for any $n \in \mathbb{Z}^+$, the product will equal 1. Substituting, we have
\begin{align*}
	E\left[ t(z) \right] &= \frac{1}{n} \log \left\{ \prod_{z = 1}^n \left(\frac{z}{n} - \frac{1}{2n}\right) \bigg/ \left(1-\frac{z}{n} + \frac{1}{2n}\right) \right\} \\
	&= \frac{1}{n} \log(1) = 0
\end{align*}
Thus, the expected value of the sufficient statistic is zero.
\end{proof}


\end{appendix}

\bibliographystyle{abbrvnat} 
\bibliography{fullbib.bib}       

\end{document}